\documentclass[11pt]{article}
\usepackage{times}
\usepackage{url, amsmath, amsthm, amssymb, epsfig, color, subfloat}
\usepackage[linesnumbered,ruled,vlined]{algorithm2e}
\usepackage{wrapfig}
\usepackage[aboveskip=-50pt]{subcaption}
\usepackage{algpseudocode}
\usepackage{times}
\usepackage{url, amsmath, amsthm, amssymb, epsfig, color, xcolor, xspace,  graphicx, epstopdf}
\usepackage{tikz}
\usetikzlibrary{calc,arrows,decorations.pathmorphing,backgrounds,positioning,fit}
\usepackage{pgfplots}
\usepackage[margin=1.0in]{geometry}
\include{inputtikz}
\include{array}
\usepackage{multirow}
\usepackage{multicol}

\usepackage{algpseudocode}

\SetCommentSty{mycommfont}

\def\etal{\emph{et al.}}
\newcounter{listcounter}

\renewcommand{\thelistcounter}{\roman{listcounter}}
\newcommand{\descr}{\begin{list}{(\thelistcounter)}
{\usecounter{listcounter}
\setlength{\rightmargin}{0mm}}}
\newtheorem{lemma}{Lemma}[section]
\newtheorem{theorem}[lemma]{Theorem}

\newtheorem{claim}[lemma]{Claim}
\newtheorem{fact}[lemma]{Fact}

\newif\ifpfsymb

\newenvironment{Pf}{\par\addvspace{6pt}\addtocounter{proof}{1}
       \noindent{\bf Proof:}}{{\ifnum\value{proof}>0
                \hfill\hbox{\mbox{}\hfill$\Box$}
		\addtocounter{proof}{-1}\par\addvspace{10pt}\else\par\fi
        }}

\makeatletter
\newcounter{algo}
\def\thealgo{\@arabic\c@algo}
\def\fps@algo{tbp}
\def\ftype@algo{1}
\def\ext@algo{loa}
\def\fnum@algo{Algorithm \thealgo}
\def\algo{\@float{algo}}
\def\endalgo{\end@float}
\makeatother

\makeatletter
\def\remlab#1{\@bsphack\if@filesw {\let\thepage\relax
   \def\protect{\noexpand\noexpand\noexpand}%
\xdef\@gtempa{\write\@auxout{\string
	   \newlabel{rem:#1}{{\thelemma}{\thepage}}}}}\@gtempa
            \if@nobreak \ifvmode\nobreak\fi\fi\fi\@esphack}
\def\deflab#1{\write\@auxout{\string
	\newlabel{def:#1}{{\thelemma}{\thepage}}}}
\makeatother

{\catcode`\@=11
\gdef\setft#1#2#3{%
\def\@oddfoot{%
{\setbox0=\hbox{#1}%
\setbox1=\hbox{#3}%
\ifdim\wd0>\wd1%
\dimen0=\wd0%
\box0\hfil#2\hfil\hbox to\dimen0{\hfil\hfil\box1}%
\else \dimen0=\wd1%
\hbox to\dimen0{\box0\hfil}\hfil#2\hfil\box1\fi%
}}}}

{\catcode`\@=11
\gdef\sethd#1#2#3{%
\def\@oddhead{\vbox{\hbox to\hsize{{#1}\hfil{#2}\hfil{#3}}%
\vspace{0.06in}%
\hbox to \hsize{\hrulefill}\vspace*{-0.09in}}}
\def\@evenhead{\@oddhead}
	}
}

\def\mysecn#1{\setcounter{equation}{0}
\section*{#1}\mark{#1}}

\def\thebibliography#1{\mysecn{References}
\addcontentsline{toc}{section}{References}\list
{[\arabic{enumi}]}{\settowidth\labelwidth{[#1]}\leftmargin\labelwidth
 \advance\leftmargin\labelsep
 \usecounter{enumi}}
 \def\newblock{\hskip .11em plus .33em minus .07em}
 \sloppy\clubpenalty4000\widowpenalty4000
 \sfcode`\.=1000\relax
 \small}

\def\complaint#1{}
\def\withcomplaints{
\newcounter{mycomplaints}
\def\complaint##1{\refstepcounter{mycomplaints}%
\ifhmode%
\unskip%
{\dimen1=\baselineskip \divide\dimen1 by 2 %
\raise\dimen1\llap{\tiny -\themycomplaints-}}\fi%
\marginpar{\tiny [\themycomplaints]: ##1}}%
}

\newcounter{printertype}
\def\figprint#1{
        \ifcase \theprintertype

		\begin{center}
                 \input{#1}
		\end{center}
              \or
                 \centerline{\psfig{figure=#1.ps}}
              \else
                 \vspace*{1in}
        \fi}

\makeatletter
\long\def\@myfootnotetext#1{\insert\footins{\footnotesize
    \interlinepenalty\interfootnotelinepenalty 
    \splittopskip\footnotesep
    \splitmaxdepth \dp\strutbox \floatingpenalty \@MM
    \hsize\columnwidth \@parboxrestore
   \edef\@currentlabel{\csname p@footnote\endcsname\@thefnmark}\@makemyfntext
    {\rule{\z@}{\footnotesep}\ignorespaces
      #1\strut}}}

\def\myfootnotetext{\@ifnextchar
[{\@xfootnotenext}{\xdef\@thefnmark{\thempfn}\@myfootnotetext}}

\long\def\@makemyfntext#1{\parindent 5mm #1}

\newcounter{proof}
\def\@@meqncr{\let\@tempa\relax
    \ifcase\@eqcnt \def\@tempa{& & &}\or \def\@tempa{& &}
      \else \def\@tempa{&}\fi
     \@tempa $\Box$\addtocounter{proof}{-1}
     \global\@eqnswtrue\global\@eqcnt\z@\cr}
\@namedef{meqnarray*}{\def\@eqncr{\nonumber\@seqncr}\eqnarray}
\@namedef{endmeqnarray*}{\global\@eqcnt\tw@\@@meqncr\egroup
      \global\advance\c@equation\m@ne$$\global\@ignoretrue}

\@namedef{meqnarray}{\eqnarray}
\@namedef{endmeqnarray}{\@@meqncr\egroup
        \global\advance\c@equation\m@ne$$\global\@ignoretrue}
\def\mequation{$$\global\@ignoretrue}

\@namedef{Pf*}{\Pf}
\@namedef{endPf*}{\par\addvspace{10pt}}

\makeatother

\def\Ex{\ensuremath{\textbf{E}}}

\def\D{\ensuremath{\mathcal{D}}}
\def\T{\ensuremath{\mathcal{T}}}
\def\S{\ensuremath{\mathcal{S}}}
\def\eps{\ensuremath{\epsilon}}
\def\opt{{\textsc{Opt}}}
\def\cs{\ensuremath{(3.1)}}
\def\css{\ensuremath{(2.14)}}

\newtheorem{corollary}{Corollary}[section]

\renewcommand{\Re}{\mathbb{R}}

\newcommand{\remove}[1]{}

\newcommand{\hide}[1]{#1}

\newsavebox{\smallProofsym}                            
\savebox{\smallProofsym}                               %
{
\begin{picture}(6,6)                                   %
\put(0,0){\framebox(6,6){}}                            %
\put(0,2){\framebox(4,4){}}                            %
\end{picture}                                          %
}   

\begin{document}
        
\title{Tighter Estimates  for $\eps$-nets for Disks}
\date{}
 
\author{ 
        Norbert Bus \\ Universit\'e Paris-Est, \\ Laboratoire d'Informatique Gaspard-Monge, \\ Equipe A3SI, ESIEE Paris. \\ busn@esiee.fr \and
        Shashwat Garg \\ IIT Delhi \\ garg.shashwat@gmail.com \and
                Nabil H. Mustafa \\ Universit\'e Paris-Est, \\ Laboratoire d'Informatique Gaspard-Monge, \\ Equipe A3SI, ESIEE Paris.\\ mustafan@esiee.fr \and
        Saurabh Ray \\  Computer Science, New York University, Abu Dhabi. \\ saurabh.ray@nyu.edu}
 
\maketitle
 \thispagestyle{empty}
%




\begin{abstract}

The geometric hitting set problem is one of the basic
geometric combinatorial optimization problems: given
a set $P$ of points, and a set $\D$ of geometric objects
in the plane, the goal is to compute a small-sized
subset of $P$ that hits all objects in $\D$. 
In 1994, Bronniman and Goodrich~\cite{BG95} made an important
connection of this problem to the size of fundamental combinatorial structures
called $\eps$-nets, showing that small-sized $\eps$-nets
imply approximation algorithms with correspondingly small
approximation ratios. Very recently, Agarwal-Pan~\cite{AP14}
showed that their scheme can be implemented in near-linear
time for disks in the plane. Altogether this gives
$O(1)$-factor approximation algorithms in $\tilde{O}(n)$ time
for hitting sets for disks in the plane.

This constant factor depends on the sizes of $\eps$-nets for disks;
unfortunately, the current state-of-the-art bounds  
are large -- at least $24/\eps$ and most likely  larger 
than $40/\eps$. Thus the approximation factor of the Agarwal-Pan algorithm
ends up being more than $40$. The best lower-bound
is $2/\eps$, which follows from the Pach-Woeginger construction~\cite{PW}
for halfspaces in two dimensions. Thus there is a large gap
between the best-known upper and lower bounds.
Besides being of independent interest, finding
precise bounds is important since this immediately
implies an improved linear-time algorithm for the hitting-set problem.

The main goal of this paper is to improve the upper-bound
to $13.4/\eps$ for disks in the plane.
The proof is constructive, giving a simple algorithm that  uses only Delaunay triangulations.
We have implemented the algorithm, which is available as a public open-source module.
Experimental results show that the
sizes of $\eps$-nets for a variety of data-sets is lower,
around $9/\eps$.


\end{abstract} 


\section{Introduction}
The minimum hitting set problem is one of the most fundamental combinatorial optimization
problems: given a range space $(P, \D)$ consisting of a set $P$ and a set $\D$ of subsets of $P$ called the \emph{ranges},   the task is to compute the smallest subset $Q \subseteq P$ that has a non-empty 
intersection with each of the ranges in $\D$. This problem is  strongly NP-hard.
If there are no restrictions on the set system $\D$, then it is known that it is NP-hard 
to approximate the minimum hitting set within a logarithmic factor of the optimal~\cite{RS97}.
The problem is NP-complete even for the case where each range has exactly two points since this problem is equivalent to the vertex cover problem which is known to be NP-complete ~\cite{Karp72,GJ79}. 
A natural occurrence of the hitting set problem occurs when the range space $\D$ 
is derived from geometry -- e.g., given a set $P$ of $n$ points in $\Re^2$, 
and a set $\D$ of $m$ triangles containing points of $P$, compute the minimum-sized subset of $P$ that 
hits all the triangles in $\D$. 
Unfortunately, for most natural geometric range spaces, computing
the minimum-sized hitting set remains NP-hard. For example, even the (relatively) simple case where $\D$ is a
set of unit disks in the plane is strongly NP-hard~\cite{HM87}. 
Therefore fast algorithms for computing provably good approximate hitting sets for geometric
range spaces have been intensively studied for the past three decades (e.g., see 
the two recent PhD theses on this topic~\cite{F12, G11}). 

The case studied in this paper -- hitting sets for disks in the plane -- 
has been the subject of a long line of research. 
The case when all the disks have the same radius is easier, and has
been studied in a series of works:
C\u{a}linsecu \etal~\cite{CMWZ04} proposed a $108$-approximation
algorithm, which was subsequently improved by Ambhul \etal~\cite{AEMN06} to $72$.
Carmi \etal~\cite{CKL07} further improved that to a $38$-approximation
algorithm, though with the running time of $O(n^6)$.
Claude \etal~\cite{CDDDFLNS10} were able to achieve a $22$-approximation
algorithm running in time $O(n^6)$. More recently
Fraser \etal~\cite{DFLN12} presented a $18$-approximation algorithm
in time $O(n^2)$.

So far, besides ad-hoc approaches, there are two systematic lines
along which all progress on the hitting-set problem for geometric
ranges has relied on: rounding via $\eps$-nets, and local-search.
The local-search approach starts
with any hitting set $S \subseteq P$, and repeatedly decreases the size of $S$, if possible,
by replacing $k$ points of $S$ with $\leq k-1$ points of $P \setminus S$. Call
such an algorithm a $k$-local search algorithm.
It has been shown~\cite{MR10} that a $k$-local search algorithm for the hitting
set problem for disks in the plane gives a PTAS. Unfortunately the running time of their algorithm to compute a $(1+\eps)$-approximation is $O(n^{O(1/\eps^2)})$.
Very recently Bus \etal~\cite{BGMR15} were able to improve the analysis
and algorithm of the local-search approach to design a 
$8$-approximation running in time $O(n^{2.33})$. However, at this moment, a near-linear
time algorithm based on local-search seems beyond reach.
We currently do not even know how to compute the most trivial case, namely when $k=1$,
of local-search in near-linear time: given the set of disks $D$, and a set
of points $P$, compute a \emph{minimal} hitting set in $P$ of $D$.

  

\paragraph{Rounding via $\eps$-nets.}
Given a range space $(P, \D)$ and a parameter $\eps>0$, an $\eps$-net is a subset $S \subseteq P$ such that
$D \cap S \neq \emptyset$ for all $D \in \D$ with $|D \cap P| \geq \eps n$. The famous ``$\eps$-net theorem'' 
of Haussler and Welzl~\cite{HW87} states that for range spaces with  VC-dimension $d$,
there exists an $\eps$-net of size $O(d/\eps \log d/\eps)$ (this bound was later improved to $O(d/\eps \log 1/\eps)$,
which was shown to be optimal in general~\cite{PA95, M02}).  
\hide{
Sometimes, weighted versions of the problem are considered in which each $p\in P$ has some positive weight associated with it so that the total weight of all elements of $P$ is $1$. The weight of each range is the sum of the weights of the elements in it. The aim is to hit all ranges with weight more than $\epsilon$.
The condition of having finite $VC$-dimension is satisfied        
by many geometric set systems: disks, half-spaces, $k$-sided polytopes, $r$-admissible set of
regions etc. in $\Re^d$. }
For certain range spaces, one can even show the existence of $\eps$-nets of size $O(1/\eps)$ -- an important case being for disks in $\Re^2$~\cite{PR08}.
 
In 1994, Bronnimann and Goodrich~\cite{BG95} proved the following interesting
connection between the hitting-set problem, and $\eps$-nets: let $(P, \D)$ be a range-space 
for which we want to compute a minimum hitting set. If one can compute an $\epsilon$-net of size $c/\epsilon$ for the $\epsilon$-net problem for $(P,\D)$ in polynomial time, 
then one can compute
a hitting set of size at most $c \cdot \opt$ for $(P,\D)$, where
$\opt$ is the size of the optimal (smallest) hitting set, in polynomial time. 
\hide{A shorter, simpler proof   
was given by Even \emph{et al.}~\cite{ERS05}. }
\hide{Both these proofs construct
an assignment of weights to points in $P$ such that the total weight of each
range  $D \in \D$ (i.e., the sum of the weights of the points in $D$) is at least $(1/\opt)$-th fraction of the total weight. Then
a $(1/\opt)$-net with these weights is a hitting set.}
Until very recently, the best such rounding algorithms had running times 
of $\Omega(n^2)$, and it had been a long-standing open problem to compute a $O(1)$-approximation
to the hitting-set problem for disks in the plane in near-linear time.
In a recent break-through, Agarwal-Pan~\cite{AP14} presented an algorithm
that is able to do the required rounding efficiently   for a broad set of geometric objects.
In particular, they are able to get the first near-linear algorithm for computing
$O(1)$-approximations for hitting sets for disks.

\paragraph{Bounds on $\eps$-nets.} 
The result of Agarwal-Pan~\cite{AP14} opens the way, for the first time,
for near linear-time algorithms for the geometric hitting set problem. The catch
is that the approximation factor depends on the sizes of $\eps$-nets for disks;
despite over 7 different proofs of $O(1/\eps)$-sized $\eps$-nets for disks, the
precise bounds are not very encouraging. The paper containing the earliest proof,  Matousek \etal~\cite{MSW90},  was over twenty-two years ago and thus summarized their result:

\textit{``Note that in principle the $\eps$-net construction presented in this paper can be
transformed into a deterministic algorithm that runs in polynomial time, $O(n^3)$
at worst. However, we certainly would not advocate this algorithm 
as being practical. We find the resulting constant of proportionality also not
particularly flattering.''}~\cite{MSW90}

So far, the best constants for the $\epsilon$-nets come from the proofs in \cite{PR08} and \cite{HPKSS14}. The latter paper presents five proofs for the existence of linear size $\epsilon$-nets for halfspaces in $\Re^3$. The best constant for disks is obtained by using their first proof. A lifting of the problem of disks to $\Re^3$ gives an $\epsilon$-net problem with lower halfspaces in $\Re^3$, for
which \cite{HPKSS14} obtains a bound of $\frac{4}{\epsilon} f(\alpha)$ where $\alpha < \frac{1}{3}$ and $f(\alpha)$ is the best bound on the size of an $\alpha$-net for lower halfspaces in $\Re^3$. Using the lower bound of \cite{PW} for halfspaces in $\Re^2$, $f(\alpha) \geq \lceil 2/\alpha \rceil -1 \geq 6$, although we believe that it is at least $10$ since even for $\epsilon = 1/2$, no $\epsilon$-net construction of size less than $10$ is known. Thus, the best constructions so far give a bound that is at least  $24/\epsilon$ and most likely more than $40/\epsilon$. 
Furthermore,
there is no implementation or software solution available that can even compute
such $\eps$-nets efficiently.

\subsection*{Our Contributions}





We prove new improved bounds on sizes of $\eps$-nets and present efficient algorithms
to compute such nets. 
Our approach is simple: we will show that modifications to a well-known technique
for computing $\eps$-nets -- the sample-and-refine approach of Chazelle-Friedman~\cite{CF90}
-- together with additional structural properties of Delaunay triangulations in fact results in $\eps$-nets of surprisingly low size:
 
\begin{theorem}
\label{thm:epsnets}
Given a set $P$ of $n$ points in $\Re^2$, there exists an $\eps$-net
under disk ranges of size at most $13.4/\eps$. Furthermore it can
be computed in expected time $O(n \log n)$.
\end{theorem}

A major advantage of Delaunay triangulations is that their
behavior has been extensively studied, there are many efficient
implementations available, and they exhibit good behavior for various real-world data-sets as well as random point sets.
The algorithm, using \textsc{CGAL}, is furthermore simple to implement. We have implemented it, and present
the sizes of $\eps$-nets for various real-world data-sets; the results
indicate that our theoretical analysis closely tracks the actual size of the nets.
This can additionally be seen as continuing the program for
better analysis of basic geometric tools; see, e.g.,
Har-Peled~\cite{HP00} for analysis of algorithms and 
Matousek~\cite{M98} for detailed analysis, both for a related
structure called cuttings in the plane.

Together with the result of Agarwal-Pan, this immediately implies the following:

\begin{corollary}
For any $\delta > 0$, one can compute a $(13.4+\delta)$-approximation to the minimum hitting
set for $(P, \D)$ in time $\tilde{O}(n)$.
\end{corollary}

\section{A near linear time algorithm for computing $\eps$-nets for disks in the plane}
\label{sec:epsnets}

Through a more careful analysis, we present an algorithm for computing an $\eps$-net of size $\frac{13.4}{\eps}$, running in near linear time. The method, shown in Algorithm~\ref{alg:epsnets}, computes a random sample and then solves certain subproblems involving subsets   located in pairs of Delaunay disks circumscribing adjacent triangles in the Delaunay triangulation of the random sample. The key to improved bounds is $i)$ considering edges in the Delaunay triangulation instead of faces in the analysis, and $ii)$ new improved constructions for large values of $\eps$.


Let $\Delta(abc)$ denote the triangle
defined by the three points $a$, $b$ and $c$. 
$D_{abc}$ denotes the disk through $a$, $b$ and $c$, 
while $D_{ab\overline{c}}$ denotes the halfspace defined by $a$ and $b$ not containing
the point $c$. Let $c(D)$ denote the center of the disk $D$.

Let $\Xi(R)$ be the Delaunay triangulation of a set of points $R \subseteq P$
in the plane. We will use $\Xi$ when $R$ is clear from the context.
For any triangle $\Delta \in \Xi$, let  $D_{\Delta}$
be the Delaunay disk of $\Delta$, and let $P_{\Delta}$ be  the
set of points of $P$ contained in $D_{\Delta}$. Similarly, 
for any edge $e \in \Xi$, let $\Delta^1_e$ and $\Delta^2_e$ be
the two triangles in $\Xi$ adjacent to $e$, and $P_e = P_{\Delta^1_e} \bigcup P_{\Delta^2_e}$. If $e$ is on the convex-hull, then one of the triangles is taken to be
the halfspace defined by $e$ not containing $R$.

\IncMargin{1em}
\begin{algorithm}[!h!]                            
\DontPrintSemicolon
\KwData{Compute $\eps$-net, given $P$: set of $n$ points in $\Re^2$, $\eps > 0$ and $c_1$.}
 \BlankLine
 \If{$\eps n < 13$}{Return $P$}
 Pick each point $p \in P$ into $R$ independently with probability $\frac{c_1}{\eps n}$.\;
 \If{$|R| \leq c_1/2\eps$}{restart algorithm.}
 Compute the Delaunay triangulation $\Xi$ of $R$.\;
 \For{triangles $\Delta \in \Xi$}{
    Compute the set of points $P_{\Delta} \subseteq P$ in Delaunay disk $D_{\Delta}$ of $\Delta$.\;
 }
 \For{edges $e \in \Xi$}{
    Let $\Delta^1_{e}$ and $\Delta^2_{e}$ be the two triangles adjacent to $e$, $P_e = P_{\Delta^1_e} \cup P_{\Delta^2_e} $.\;
    Let $\eps' = (\frac{\eps n}{|P_e|})$ and compute a $\eps'$-net $R_{e}$ for $P_e$ depending on the cases below:\;
    \If{$\frac{2}{3} < \eps' < 1$}{compute using Lemma~\ref{lemma:2}.}
    \If{$\frac{1}{2} < \eps' \leq \frac{2}{3}$}{compute    using Lemma~\ref{lemma:10}.}
    \If{$\eps' \leq \frac{1}{2}$}{compute   recursively.}
  }
  \textbf{Return} $\left(\bigcup_{e} R_e\right) \cup R$.\;
\caption{Compute $\eps$-nets \label{alg:epsnets}}
\end{algorithm}
\DecMargin{1em}

In order to prove that the algorithm gives the desired result, the following theorems regarding the size of an $\eps$-net will be useful. Let $f(\eps)$ be the size
of the smallest $\eps$-net for any set $P$ of points in $\Re^2$ under
disk ranges. 

\begin{lemma}[\cite{AAG14}]
\label{lemma:2}
For $\frac{2}{3} < \eps < 1$, $f(\eps) \leq 2$, and can be computed in $O(n\log n)$  time.
\end{lemma}
\begin{lemma}
\label{lemma:10}
 For $\frac{1}{2} < \eps \leq \frac{2}{3}$, $f(\eps) \leq 10$ and can be computed in $O(n\log n)$ time.
\end{lemma}
\hide{
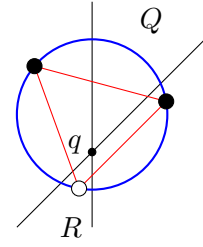
\begin{wrapfigure}{r}{0.3\textwidth}
\centering
\vspace{-0.6in}
\begin{tikzpicture}[scale=0.5]
\draw[black] (0,0) to (5,5);
\draw[black] (2,0) to (2,6);
\draw[blue, thick] (2,3) circle (2);
\node[black,left] at (2,2.2) {$q$};
\draw[black, fill=black] (2,2) circle (0.1);
\node[black,right] at (3,5.5) {$Q$};
\node[black,left] at (2,0.0) {$R$};
\draw[red] ($(2,3)+(10:2)$) to ($(2,3)+(140:2)$) to ($(2,3)+(260:2)$) -- cycle;
\draw[black, fill=black] ($(2,3)+(10:2)$) circle (0.2);
\draw[black, fill=black] ($(2,3)+(140:2)$) circle (0.2);
\draw[black, fill=white] ($(2,3)+(260:2)$) circle (0.2);
\end{tikzpicture}
\caption{Setup around $q$.}
\end{wrapfigure}
\emph{Proof.}
Divide the plane into $4$ quadrants with $2$ lines,  intersecting at a point $q$, such that each quadrant contains $n / 4$ points. Using the Ham-Sandwich theorem, this can be done in linear time~\cite{M02}. Create a $\frac{2}{3}$-net for each quadrant, using Lemma~\ref{lemma:2}. Add these $8$ points to the $\eps$-net of $P$. If $q\in P$ then add $q$ to the $\eps$-net; otherwise let $\Delta$ be the triangle in the Delaunay triangulation
of $P$ that contains the point $q$. Add the two vertices
of $\Delta$ that are in the opposite quadrants to the $\eps$-net. The resulting
size of the net is at most $10$.
Denote the quadrant without a vertex of the Delaunay triangle inside it by $Q$ and its opposite quadrant by $R$.
If a disk $D$ intersects at most $3$ quadrants and does not contain
any of the points from the $\frac{2}{3}$-net in each of those quadrants, it can contain only at most $3\cdot\frac{2}{3}\cdot\frac{n}{4}=\frac{n}{2}$ points.
On the other hand, if $D$ contains points from each of the $4$ quadrants, then it must contain points from $Q$ and $R$ that are outside of the Delaunay disk $D_{\Delta}$ of $\Delta$ (as $D_{\Delta}$  is empty of points of $P$). Then if $D$ does not contain
any of the two vertices of $\Delta$ in the opposite quadrants (already added
to the $\eps$-net), it must
pierce $D_{\Delta}$, a contradiction.
\qed
}

Call a tuple $(\{p,q\},\{r,s\})$, where $p,q,r,s \in P$, a \emph{Delaunay quadruple} if
$int(\Delta(pqr)) \cap int(\Delta(pqs)) = \emptyset$.
Define its \emph{weight}, denoted $W_{(\{p,q\},\{r,s\})}$, to be the number of points of $P$ in $D_{pqr} \cup D_{pqs}$.
Let $\T_{\leq k}$ be a set of Delaunay quadruples of $P$ of weight at most $k$ and similarly $\T_k$ denotes the set of Delaunay quadruples of weight exactly $k$.
Similarly, a \emph{Delaunay triple} is given by $(\{p,q\}, \{r\})$, where $p,q,r \in P$.
Define its \emph{weight}, denoted $W_{(\{p,q\},\{r\})}$, to be the number of points of $P$ in $D_{pqr} \cup D_{pq\overline{r}}$.
Let $\S_{\leq k}$ be a set of Delaunay triples of $P$ of weight at most $k$, and   $\S_k$ denotes the set of Delaunay triples of weight exactly $k$.

One can upper bound the size of $\T_{\leq k}$, $\S_{\leq k}$ and using it, we derive an upper bound on the expected number of sub-problems with a certain number of points.
 
\begin{claim}
\label{claim:cs}
$|\T_{\leq k}| \leq (e^3/9) nk^3$ asymptotically \text{ and } $|\T_{\leq k}| \leq \cs nk^3$ for $k\geq 13$.
\end{claim}%
\hide{
\begin{proof}
The proof is an application of the Clarkson-Shor technique~\cite{M02}.
Pick each point in $P$ independently with probability $p_{cs}$ to
get a random sample $R_{cs}$. Count the expected number of   edges
in the Delaunay triangulation of $R_{cs}$ in two ways. On one hand, it is simply
less than $3E[|R_{cs}|] = 3np_{cs}$. 
On the other hand, it is: 
\begin{eqnarray*}
3np_{cs} &\geq& \Ex [\text{Number of Delaunay edges in } R_{cs}] =  \sum_{p,q \in P} \Pr[ \{p,q\} \text{ is a Delaunay edge of } R_{cs} ] \\
&\geq& \sum_{p,q \in P} \sum_{r,s \in P} \Pr[ (D_{pqr} \cup D_{pqs}) \cap R_{cs} = \emptyset ] \qquad \text{ (disjoint events) } \\
&\geq& \sum_{ (\{p,q\},\{r,s\}) \in \T_{\leq k} } \Pr[ (D_{pqr} \cup D_{pqs}) \cap R_{cs} = \emptyset ]\\
&\geq & \sum_{ (\{p,q\},\{r,s\}) \in \T_{\leq k} } p_{cs}^4 \cdot (1-p_{cs})^k
= |\T_{\leq k}| \cdot p_{cs}^4 \cdot (1-p_{cs})^k \\
\end{eqnarray*}
Therefore $| \T_{\leq k} | \leq 3np_{cs}/( p_{cs}^4(1-p_{cs})^k )$ and a simple calculation gives that setting $p_{cs}=\frac{3}{k+3}$ minimizes the right hand side. Then $| \T_{\leq k} | \leq 3n\frac{3}{k+3}/( (\frac{3}{k+3})^4(1-\frac{3}{k+3})^k )=nk^3\frac{1}{9}(1+\frac{3}{k})^{k+3}$, and the claim follows. 
\end{proof}
}

\begin{claim}
\label{claim:cs2}
$|\S_{\leq k}| \leq (e^2/4) nk^2$ asymptotically \text{ and } $|\S_{\leq k}| \leq \css nk^2$ for $k\geq 13$.
\end{claim}
\begin{proof}
Pick each point in $P$ independently with probability $p_{cs}$ to get a random
sample $R_{cs}$. Count the expected number of edges
in the Delaunay triangulation of $R_{cs}$ that lie on the boundary of the
Delaunay triangulation, i.e.,   adjacent to exactly one triangle, in two ways.
On one hand, it is exactly the number of edges in the convex-hull
of $R_{cs}$, therefore  at most $E[|R_{cs}|] = np_{cs}$. Counted another way, it is:
\begin{eqnarray*} 
np_{cs} &\geq& \Ex [\text{Number of boundary Delaunay edges in } R_{cs}] =  \sum_{p,q \in P} \Pr[ \{p,q\} \text{ is a boundary Delaunay edge of } R_{cs} ] \\
&\geq& \sum_{p,q \in P} \sum_{r \in P} \Pr[ (D_{pqr} \cup D_{pq\overline{r}}) \cap R_{cs} = \emptyset ] \qquad \text{ (disjoint events) } \\
&\geq& \sum_{ (\{p,q\},\{r\}) \in \S_{\leq k} } \Pr[ (D_{pqr} \cup D_{pq\overline{r}}) \cap R_{cs} = \emptyset ]\\
&\geq & \sum_{ (\{p,q\},\{r\}) \in \S_{\leq k} } p_{cs}^3 \cdot (1-p_{cs})^k
= |\S_{\leq k}| \cdot p_{cs}^3 \cdot (1-p_{cs})^k \\
\end{eqnarray*}
Setting $p_{cs} = \frac{2}{k+2}$ gives the required result.
\end{proof}

\begin{claim}
$$ \Ex \Big[ | \{ e \in \Xi \ \ | \ \  k_1 \eps n \leq |P_e| \leq k_2 \eps n \} | \Big] \leq  \frac{\cs c_1^3}{\eps e^{k_1c_1}} (k_1^3c_1+3.7 k_2^2) \ \text{\  if \ } \eps n \geq 13.$$
\end{claim}
\hide{
\begin{proof}
The crucial observation is that two points $\{p, q\}$ form an edge in $\Xi$ with two adjacent
triangles $\Delta(pqr), \Delta(pqs) \in \Xi$ \textbf{iff} $\{p,q,r,s\} \subseteq R$  and
none of the points of $P$ in $D_{pqr} \cup D_{pqs}$ are picked in $R$ (i.e, the points $p,q,r,s$ form the Delaunay tuple  $( \{p,q\}, \{r,s\})$). 
Or $\{p,q\}$ form an edge on the convex-hull of $\Xi$ with one adjacent
triangle $\Delta(pqr)$ \textbf{iff} $\{p,q,r\} \subseteq R$ and none of
the points of $P$ in $D_{pqr} \cup D_{pq\overline{r}}$ are picked in $R$.

Let $\chi_{(\{p,q\}, \{r,s\})}$ be the random variable that is $1$ iff
$\{p,q\}$ form an edge in $\Xi$ and their two adjacent triangles
are $\Delta(pqr)$ and $\Delta(pqs)$. Let $\chi_{(\{p,q\}, \{r\})}$ be the random variable that is $1$ iff
$\{p,q\}$ form an edge in $\Xi$ with exactly one adjacent
triangle $\Delta(pqr)$.
Noting that every edge in $\Xi$ must come from either a Delaunay quadruple or a Delaunay
triple, 
\begin{eqnarray*}
 \Ex[ |\{ e \ \ | \ \ k_1 \eps n \leq |P_e| \leq k_2 \eps n \} |]
 &=& \sum_{\substack{p,q,r,s \in P \\ k_1 \eps n \leq W_{(\{p,q\},\{r, s\})} \leq  k_2 \eps n }} \Pr[ \chi_{(\{p,q\}, \{r,s\})} = 1] \ \  + \\ 
 &&\sum_{\substack{p,q,r \in P \\ k_1 \eps n \leq W_{(\{p,q\},\{r\})} \leq k_2 \eps n }} \Pr[ \chi_{(\{p,q\}, \{r\})} = 1]
\end{eqnarray*} 
The second term is asymptotically smaller, so we   bound it somewhat loosely:
\begin{eqnarray*}
\sum_{\substack{p,q,r \in P \\ k_1 \eps n \leq W_{(\{p,q\},\{r\})} \leq k_2 \eps n }} \Pr[ \chi_{(\{p,q\}, \{r\})} = 1] 
&\leq& \sum_{\substack{p,q,r \\ k_1 \eps n \leq W_{(\{p,q\},\{r\})} \leq k_2 \eps n}} (c_1/\eps n)^3 (1-c_1/\eps n)^{W_{(\{p,q\},\{r\})}} \\
&\leq& |\S_{\leq k_2 \eps n}| \cdot (c_1/\eps n)^3 (1-c_1/\eps n)^{k_1 \eps n} \\
&\leq& \css n (k_2 \eps n)^2 \cdot (c_1 /\eps n)^3 \cdot e^{-c_1k_1} = \frac{\css k_2^2 c^3_1}{\eps e^{c_1 k_1}}.
\end{eqnarray*}
Now we carefully bound the first term:
\begin{eqnarray*}
\sum_{\substack{p,q,r,s \in P \\ k_1 \eps n \leq W_{(\{p,q\},\{r, s\})} \leq  k_2 \eps n }} \Pr[ \chi_{(\{p,q\}, \{r,s\})} = 1]
&\leq& \sum_{i = k_1 \eps n }^{ k_2 \eps n} \sum_{\substack{p,q,r,s \\ W_{(\{p,q\},\{r, s\})} = i}} \Pr[ \chi_{(\{p,q\}, \{r,s\})} = 1] \\
&\leq& \sum_{i = k_1 \eps n }^{ k_2 \eps n} \sum_{\substack{p,q,r,s\\ W_{(\{p,q\},\{r, s\})} = i}} (c_1/\eps n)^4 (1-c_1/\eps n)^{i} \\
 & \leq &  \sum_{i = k_1 \eps n }^{ k_2 \eps n}  |\T_{i}|(c_1/\eps n)^4 (1-c_1/\eps n)^{i}  
\end{eqnarray*}
As the above summation is exponentially decreasing as a function of $i$, it is
maximized when $|\T_{i_0}|= \max |\T_{\leq i_0}|$ where $i_0 = k_1\eps n$, and $|\T_i|=\max|\T_{\leq i}|-\max|\T_{\leq i-1}|$ and 
so on. Using Claim~\ref{claim:cs} we obtain:
\begin{eqnarray*}
 & \leq & |\T_{\leq k_1\eps n}|\cdot (c_1/\eps n)^4 (1-c_1/\eps n)^{k_1\eps n} + \sum_{i = k_1 \eps n + 1}^{ k_2 \eps n} (|\T_{\leq i}|-|\T_{\leq i-1}|)\cdot (c_1/\eps n)^4 (1-c_1/\eps n)^{i} \\
 & \leq & \cs n (k_1\eps n)^3\cdot (c_1/\eps n)^4 (1-c_1/\eps n)^{k_1\eps n} + \sum_{i = k_1 \eps n + 1}^{ k_2 \eps n} \cs n\cdot 3 i^2\cdot (c_1/\eps n)^4 (1-c_1/\eps n)^{i} \\
 & \leq & \cs \frac{k_1^3c_1^4e^{-k_1c_1}}{\eps}  + \cs \frac{3k_2^2 c_1^4}{\eps^2 n} \sum_{i = k_1 \eps n + 1}^{ k_2 \eps n}  (1-c_1/\eps n)^{i} \\
 & \leq & \cs \frac{k_1^3c_1^4e^{-k_1c_1}}{\eps}  + \cs \frac{3k_2^2 c_1^4}{\eps^2 n} \frac{(1-c_1/\eps n)^{k_1\eps n}}{c_1/\eps n} \leq  \frac{\cs c_1^3}{\eps e^{k_1c_1}} (k_1^3c_1+3k_2^2). \\
\end{eqnarray*}
The proof follows by summing up the two terms.
\end{proof}
}

Using the above facts we can prove the main result.

\begin{lemma}
Algorithm~\textsc{Compute $\eps$-net} computes an $\eps$-net of
expected size $13.4/\epsilon$.
\label{lemma:epsnetproof}
\end{lemma}
\begin{wrapfigure}{L}{0.3\textwidth}
\vspace{-0.2in}
\includegraphics[width=5cm]{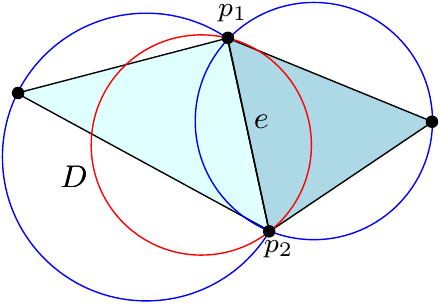}
\vspace{-0.2in}
\end{wrapfigure}
\emph{Proof.}
First we show that the algorithm computes an $\eps$-net. 
Take any disk
$D$ with center $c$ containing $\eps n$ points of $P$, and not hit by the initial random
sample $R$. Increase its radius while keeping its center $c$ fixed until
it passes through a point, say $p_1$ of $R$. Now further expand the disk
by moving $c$ in the direction $\vec{p_1c}$ until its boundary passes through a second
point $p_2$ of $R$. 
The edge $e$ defined by $p_1$ and $p_2$
belongs to $\Xi$, and the 
two extreme disks in the pencil of empty disks through $p_1$ and $p_2$
are the disks $D_{\Delta^1_e}$ and $D_{\Delta^2_e}$. Their union covers $D$,
and so $D$ contains $\eps n$ points out of the set $P_e$. 
Then the net $R_e$ computed for $P_e$ must hit $D$,
as $\eps n = (\eps n/|P_e|) \cdot |P_e|$.

For the expected size, clearly, if $\eps n < 13$ then the returned set is an $\eps$-net of size $\frac{13}{\eps}$.
Otherwise we can calculate the expected number of points added to the $\eps$-net during solving the sub-problems. We simply group them by the number of points in them. 
Set $E_i = \{e \ | \ 2^i \eps n \leq |P_e| < 2^{i+1} \eps n\}$, and let us denote the size of the $\eps$-net returned by our algorithm with $f'(\eps)$. Then
\begin{eqnarray*}
\Ex\left[f'(\eps)\right] = \Ex[|R|] + \Ex \Big[|\bigcup_{e \in \Xi} R_e| \Big] &=& \frac{c_1}{\eps} + \Ex [|\{e \ | \ \eps n \leq |P_e| < 3 \eps n/2\}|] \cdot f(2/3) \\
&& + \Ex [|\{e \ | \ 3\eps n/2 \leq |P_e| < 2 \eps n\}|] \cdot f(1/2) \\
&& + 
\sum_{i=1} \Ex \left[ \sum_{e \in E_i} f'\left(\frac{\eps n}{|P_e|}\right) \right] \\
\end{eqnarray*}
Noting that $\Ex [ \sum_{e \in E_i} f'(\frac{\eps n}{|P_e|}) \ | \ |E_i| = t ] \leq t \Ex [f'(1/2^{i+1})]$, we get 
$$ \Ex \left[ \sum_{e \in E_i} f'\left(\frac{\eps n}{|P_e|}\right) \right] 
= \Ex  \left[ \Ex [  \sum_{e \in E_i} f'\left(\frac{\eps n}{|P_e|}\right) | E_i ]  \right]
\leq \Ex  \left[ |E_i| \cdot \Ex [ f'(1/2^{i+1})] \right] 
= \Ex  [|E_i|] \cdot \Ex [f'(1/2^{i+1})] $$
as $|E_i|$ and $f'(\cdot)$ are independent. As $\eps' = \frac{\eps n }{|P_e|} > \eps$, by induction, assume $\Ex [f'\left(\eps'\right)] \leq \frac{13.4}{\eps'}$. Then
\begin{eqnarray*}
\Ex\left[f'(\eps)\right] &\leq& \frac{c_1}{\eps} + \frac{\cs \cdot c_1^3(c_1+8.34)}{\eps e^{c_1}} \cdot 2  +  \frac{\cs \cdot c_1^3((3/2)^3c_1+14.8)}{\eps e^{3c_1/2}} \cdot 10  \\
&& + \sum_{i} \frac{\cs \cdot c_1^3(2^{3i}c_1+3.7\cdot 2^{2i+2})}{\eps e^{c_1 2^i}} \cdot 13.4 \cdot 2^{i+1}    \leq \frac{13.4}{\eps}\\
\end{eqnarray*}
 by setting $c_1 = 12$.  
\qed

Finally, we bound the expected running time of the algorithm.

\begin{lemma}
Algorithm~\textsc{Compute $\eps$-net} runs in expected time $O(n \log n)$.
\end{lemma}
\hide{
\begin{proof}
Note that $\Ex[|R|] = c_1/\eps$.
First we bound the expected total size of all the sets $P_e$:
\begin{eqnarray*}
\Ex \Big[ \sum_{e \in \Xi} |P_e| \Big] &\leq&  \Ex [ | \{e \ \ | \ \ 0 \leq |P_e| <  \eps n\}|] \cdot \eps n +  \sum_{i=0} \Ex [|\{e \ \ | \ \ 2^i \eps n \leq |P_e| < 2^{i+1} \eps n\}|] \cdot 2^{i+1}\eps n \\
&\leq& O(\frac{\eps n}{\eps}) +   \sum_{i=0} O\left( \frac{(2^{i})^3}{\eps e^{2^ic_1}} \right) \cdot 2^{i+1} \eps n = O(n),
\end{eqnarray*}
as the last summation is a geometric series. 
This implies that the expected total number of incidences between points in $P$,
and Delaunay disks in $\Xi$ is $O(n)$. The Delaunay triangulation of $R$ can be
computed in expected time $O(1/\eps \log 1/\eps)$. 
Steps 5-6 compute, for each
Delaunay disk $D \in \Xi$, the list of points contained in $D$. This can be computed in $O(n \log 1/\eps)$ time by instead finding, for each $p \in P$, the list of Delaunay disks in $\Xi$
containing $p$, as follows. First do point-location in $\Xi$ to locate the triangle $\Delta$
containing $p$, in expected time $O(\log 1/\eps)$. Clearly $D_{\Delta}$ contains $p$. Now starting from $\Delta$, do a breadth-first search
in the dual planar graph of the Delaunay triangulation to find the maximally connected
subset of triangles (vertices in the dual graph) whose Delaunay disks contain $p$. As each
vertex in the dual graph has degree at most $3$, this takes time proportional to the
discovered list of triangles, which as shown earlier is $O(n)$ over all $p \in P$. The correctness follows from the following:

\begin{fact}
Given a Delaunay triangulation $\Xi$ on $R$ and any point $p \in \Re^2$, the set of triangles in $\Xi$ whose Delaunay disks contain $p$ form a connected sub-graph in the dual graph to $\Xi$.
\end{fact}
\begin{proof}
This can be seen by lifting $P$ to $\Re^3$ via the Veronese mapping, where it follows from the fact that the faces of a convex polyhedron that are  visible from any exterior point are connected.
\end{proof}

Note that by the $\eps$-net theorem, the probability of restarting
the algorithm (lines 4-5) at any call is at most a constant. Therefore it is re-started
expected at most a constant number of times, and so the expected running time, denoted by $T(n)$:
\begin{eqnarray*}
 \Ex[ T(n) ] &=& O(1/\eps\log 1/\eps)+ O(n\log 1/\eps)+ \sum_{e \in \Xi} \Ex[ T(|P_e|) ] \leq  O(n\log 1/\eps)+  \sum_{e \in \Xi} \Ex[ T(|P_e|) ] \\
\end{eqnarray*}
Similarly to previous calculations we have that 
\begin{eqnarray*}
\Ex[ T(n) ] &\leq& O(n\log 1/\eps) + \frac{\cs \cdot c_1^3(c_1+8.34)}{\eps e^{c_1}} \cdot O(3\eps n/2 \log( 3 \eps n/2))  \\
&& + \frac{\cs \cdot c_1^3((3/2)^3c_1+14.8)}{\eps e^{3c_1/2}} \cdot O(2\eps n \log( 2 \eps n)) \\
 && + \sum_{i=1} \frac{\cs \cdot c_1^3(2^{3i}c_1+3.7\cdot 2^{2i+2})}{\eps e^{c_1 2^i}} \cdot \Ex[ T(2^{i+1} \eps n ) ] \\
 && \leq d n\log n + \sum_{i=1} \frac{\cs \cdot c_1^3(2^{3i}c_1+3.7\cdot 2^{2i+2})}{\eps e^{c_1 2^i}} \cdot \Ex[ T(2^{i+1} \eps n ) ]
\end{eqnarray*}
for a constant $d$ coming from the constants above, as well as in Delaunay triangulation, point-location and list-construction computations.
Setting $\Ex[T(k)]=c k\log k$ satisfies the above inequality for $c \geq 2d$, since
\begin{eqnarray*}
 \Ex[T(n)] &\leq& d n\log n + \sum_{i=1} \frac{\cs \cdot c_1^3(2^{3i}c_1+3.7\cdot 2^{2i+2})}{\eps e^{c_1 2^i}} \cdot c (2^{i+1} \eps n )\log(2^{i+1} \eps n ) \\
 &\leq& d n\log n + (c n \log n) \sum_{i=1} \frac{2^{i+1} \cdot \cs \cdot 12^3(2^{3i}\cdot 12+3.7\cdot 2^{2i+2})}{e^{12 \cdot 2^i}}  \\
  &\leq& d n \log n + c n \log n \cdot \frac{1}{2} \leq  c n\log n \text{, \ for \ \ $c \geq 2d$}.
\end{eqnarray*}

\end{proof}
}

\section{Implementation and Experiments}
\label{sec:experiments}

In this section we present experimental results for our algorithm running on a machine equipped with an Intel Core i7 870 processor with 4 cores each running at $2.93$ GHz and with $16$ GB main memory. 
All our implementations are single threaded in order to have a fair comparison.
For nearest-neighbors and Delaunay triangulations, we use  the well-known
geometry library \textsc{CGAL}. It computes Delaunay triangulations in
expected $O(n \log n)$ time. Instead of computing centerpoints, we will recurse
for all values of $\eps'$; this results in simple efficient code,
at the cost of slightly larger constants.

\begin{figure}[!h]
  \centering
  \begin{tikzpicture}
   \begin{axis}[
         xlabel={$c_1$},
         ylabel={$\eps$-net size multiplied by $\epsilon$},
         ymin=5, ymax=30,
         restrict y to domain=5:1000,
         width=.50\textwidth,
         height=.35\textwidth,
         legend style={legend pos=north east,font=\tiny},
         legend cell align=left,
         ]
    \addplot[blue,mark=none,thick] coordinates {
		( 2, 82.6)
		( 3, 24.04)
		( 4, 11.58)
		( 5, 10.31)
		( 6, 8.2)
		( 7, 8.7)
		( 8, 8.56)
		( 9, 9.41)
		(10, 10.16)
		(11, 11.12)
		(12, 12.26)
		(13, 13)
		(14, 14.11)
		(15, 15.67)
		(16, 16.28)
	};
	\addplot[green,mark=none,thick] coordinates {
		( 2, 181.19)
		( 3, 24.04)
		( 4, 13.19)
		( 5, 12.82)
		( 6, 10.58)
		( 7, 8.36)
		( 8, 8.44)
		( 9, 9.01)
		(10, 9.97)
		(11, 10.95)
		(12, 11.76)
		(13, 12.93)
		(14, 14.2)
		(15, 15.2)
		(16, 15.75)
	};
	\addplot[red,mark=none,thick] coordinates {
		( 2, 3918)
		( 3, 32.76)
		( 4, 17.41)
		( 5, 9.35)
		( 6, 10.04)
		( 7, 8.21)
		( 8, 8.77)
		( 9, 9.27)
		(10, 9.6)
		(11, 10.71)
		(12, 12.11)
		(13, 12.84)
		(14, 13.56)
		(15, 14.39)
		(16, 16)
	};
	\addplot[cyan,mark=none,thick] coordinates {
		( 2, 209)
		( 3, 29.87)
		( 4, 13.31)
		( 5, 10.23)
		( 6, 11.3)
		( 7, 8.28)
		( 8, 8.35)
		( 9, 9.38)
		(10, 9.5)
		(11, 10.49)
		(12, 12.45)
		(13, 12.46)
		(14, 13.89)
		(15, 14.83)
		(16, 15.29)
	};
	\addplot[black,mark=none,thick] coordinates {
		( 2, 4072)
		( 3, 27.64)
		( 4, 16.48)
		( 5, 13.61)
		( 6, 11.16)
		( 7, 8.56)
		( 8, 8.66)
		( 9, 9.46)
		(10, 9.65)
		(11, 11.36)
		(12, 12.75)
		(13, 13.01)
		(14, 13.91)
		(15, 14.85)
		(16, 16.63)
	};
	\addplot[yellow,mark=none,thick] coordinates {
    	( 2, 7939)
		( 3, 34.22)
		( 4, 17.97)
		( 5, 12.04)
		( 6, 7.58)
		( 7, 8.23)
		( 8, 8.43)
		( 9, 9.61)
		(10, 10.74)
		(11, 11.15)
		(12, 11.98)
		(13, 12.84)
		(14, 13.87)
		(15, 15.58)
		(16, 16.27)
	};
	\legend{MOPSI,KDDCU,Europe,Uniform,Gauss9,Birch3}
   \end{axis}
  \end{tikzpicture}
  \vspace{-0.3cm}
  \caption{$\epsilon$-net size multiplied by $\epsilon$ for $4$ sets, $\epsilon=0.01$.}
  \label{fig:epsnet_c} 
\end{figure}
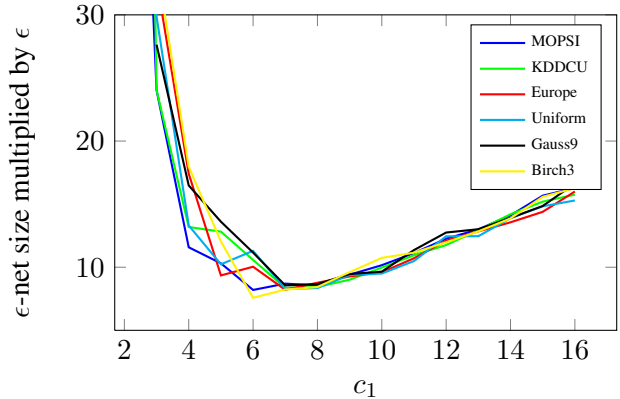

In order to empirically validate the size of the $\epsilon$-net obtained by our random sampling algorithm we have utilized several datasets in \cite{dataset}. 
The \emph{MOPSI Finland} dataset contains  $13467$ locations of users in Finland. 
The \emph{KDDCUP04Bio} dataset contains the first $2$ dimensions of a protein dataset  with $145,751$ entries. 
The \emph{Europe} and \emph{Birch3}  datasets  have $169,308$ and $100,000$ entries respectively.
We have created two random data sets \emph{Uniform} and \emph{Gauss9} with $50,000$ and $90,000$ points. 
The former is sampled from a uniform distribution while the latter is sampled from $9$ different gaussian distributions whose means and covariance matrices are randomly generated. 
Setting the probability for random sampling to $\frac{12}{\epsilon\cdot n}$ results in approximately $\frac{12}{\epsilon}$ sized nets for nearly all datasets, as expected
by our analysis. We note however, that in practice setting $c_1$ to $7$ gives smaller size $\eps$-nets, of size around $\frac{9}{\eps}$. See Figure~\ref{fig:epsnet_c} for the dependency of the net size on $c_1$ while setting $\epsilon$ to $0.01$. In Table~\ref{tab:size} we list $\epsilon$-net sizes for different values of $\epsilon$ while setting $c_1$ to $12$.

\begin{table}[!h]
 \centering
   \begin{tabular}{|c|r|r|r|r|}
	\hline 
	Dataset &  \multicolumn{4}{c|}{$\epsilon$-net size}  \\ 
	\hline 
	              & $\epsilon=0.2$   & $\epsilon=0.1$  & $\epsilon=0.01$  & $\epsilon=0.001$   \\ 
    \hline
    \emph{MOPSI Finland} & $83$  & $128$  & $1226$ & $12011$   \\ 
	\hline 
    \emph{KDDCUP04Bio}   & $55$  & $118$ & $1176$  & $11902$  \\ 
	\hline 
    \emph{Europe}        & $69$  & $119$  & $1205$ & $12043$  \\ 
	\hline 
    \emph{Birch 3}       & $58$  & $125$  & $1198$ & $11878$   \\ 
	\hline
    \emph{Uniform}       & $70$  & $109$  & $1245$ & $12034$   \\ 
	\hline
    \emph{Gauss9}       & $58$  & $120$  & $1275$ & $12011$   \\ 
	\hline
	\end{tabular} 
 \caption{ $\epsilon$-net sizes for various point sets, $c_1=12$.}
 \label{tab:size}
\end{table}

\section{Conclusion}

In this paper we have improved upon the constants in the previous construction
of $\eps$-nets for disks in the plane. Our method gives an efficient
practical algorithm for computing such $\eps$-nets, which we have implemented
and tested on a variety of data-sets. We conclude with a list of open problems:
\begin{itemize}
\item Currently the best known lower-bound is the $2/\eps$ bound for halfspaces
in $\Re^2$. It remains an interesting question to improve this lower-bound, or
improve the upper-bounds given in this paper.
\item Currently the algorithm of Agarwal and Pan~\cite{AP14} uses a number
of heavy tools (dynamic range reporting, dynamic approximate range counting)
that hinders an efficient and practical implementation of their algorithm. It would
be considerable progress to derive a more practical method with provable guarantees.
\end{itemize}

\newpage

\bibliographystyle{plain}
\bibliography{hittingsetepsnets}{}

\end{document}